\documentclass[reqno,a4paper,11pt]{amsart}
\usepackage{enumerate,amsmath,amssymb,stmaryrd}
\usepackage[mathscr]{eucal}
\usepackage{graphicx}
\usepackage{color}
\usepackage{multicol}
\usepackage{tikz}
\usepackage{xypic}
\usepackage{caption}
\usepackage{subcaption}
\usepackage{xcolor}
\usepackage{comment}
\usepackage{esint}
\usetikzlibrary{calc}
\usepackage{tikz-cd}
\usepackage{cancel}
\definecolor{bello}{RGB}{130, 45, 220} 
\usepackage[scr=esstix]   
           {mathalpha}
\usepackage[colorlinks=true, allcolors=black]{hyperref}

\newtheorem{theorem}{Theorem}[section]
\newtheorem{conj}[theorem]{Conjecture}
\newtheorem{lemma}[theorem]{Lemma}

\newtheorem{corollary}[theorem]{Corollary}
\theoremstyle{definition}
\newtheorem{definition}{Definition}[section]

\theoremstyle{remark}
\newtheorem*{remark}{Remark}

\newcommand{\red}{\textcolor{red}}

\renewcommand{\S}{\Sigma}
\renewcommand{\l}{\lambda}

\newcommand{\p}{\partial}
\newcommand{\R}{\mathbb{R}}

\newcommand{\tf}{\mathring}

\newcommand\Ric{{\rm Ric}}
\newcommand\tr{{\rm tr}}
\newcommand\dv{{\rm div}}

\newcommand{\LY}{\rm LY}
\newcommand{\BY}{\rm BY}
\renewcommand{\a}{\alpha}

\makeatletter
\def\l@subsection{\@tocline{2}{0pt}{2.5em}{5em}{}}
\makeatother

\begin{document}

\title[On energy and its positivity with an expanding flat de Sitter background]{On energy and its positivity in spacetimes with an expanding flat de Sitter background
}

\author{Rodrigo Avalos}
\address[Rodrigo Avalos]{Mathematics Department, Eberhard Karls Universit\"{a}t T\"{u}bingen}
\email{rodrigo.avalos@mnf.uni-tuebingen.de}

\author{Eric Ling}
\address[Eric Ling]{Copenhagen Centre for Geometry and Topology (GeoTop), Department of Mathematical Sciences, University of Copenhagen, Denmark
}
\email{el@math.ku.dk}

\author{Annachiara Piubello$^{\ast}$}\thanks{$^{\ast}$Author to whom any correspondence should be addressed.}
\address[Annachiara Piubello]{Copenhagen Centre for Geometry and Topology (GeoTop), Department of Mathematical Sciences, University of Copenhagen, Denmark
}
\email{anpi@math.ku.dk}

\begin{abstract}
The positive energy theorems are a fundamental pillar in mathematical general relativity. Originally proved by Schoen--Yau and later Witten, these theorems were established for asymptotically flat manifolds where the metric tends to the standard Euclidean metric and whose second fundamental form decays to zero at infinity. This ansatz on the metric and second fundamental form is motivated by the desire to model an isolated gravitational system with a Minkowski space background for the spacetime. However, actual astrophysical massive objects are not truly isolated but rather exist within an expanding cosmological universe, where the second fundamental form is umbilic. With this in mind, we seek a notion of energy for initial data sets with an umbilic second fundamental form.
In this work, we present a definition of energy in such an expanding cosmological setting. Instead of Minkowski space, we take de Sitter space as the background spacetime, which, when written in flat-expanding coordinates, is foliated by umbilic hypersurfaces each isometric to Euclidean 3-space. This cosmological setting necessitates a quasi-local energy definition, as the presence of a cosmological horizon in de Sitter space obstructs a global one. We define energy in this quasi-local setting by adapting the Liu--Yau energy to our framework and establish positivity of this energy for certain bounded values of the cosmological constant.
\end{abstract}


\maketitle



\section{Introduction}

The question we want to address is how to describe energy in isolated gravitational systems within an expanding universe. In physics, isolated gravitational systems are meant to represent localized matter sources which generate a gravitational field decaying as one moves further away from these sources. Traditionally, such systems have been described in terms of spacetimes which, in some meaningful sense, asymptote to Minkowski space $\R^{3,1}$ as one moves to space-like infinity. This can be captured by considering solutions to the Einstein equations arising by the evolution of initial data which asymptote to the canonical initial data for $\R^{3,1}$, i.e., a totally geodesic spacelike hyperplane, as we move to space-like infinity. The physical intuition behind this choice arises by considering $\R^{3,1}$ as the \emph{vacuum} solution to the Einstein equations that one would approach by moving continuously to exceedingly weaker gravitational fields created by a localized matter distribution. Nevertheless, observational evidence of the accelerated expansion of the universe strongly supports a cosmological model with a positive cosmological constant $ \Lambda > 0$, where the corresponding maximally symmetric cosmological $\Lambda$-vacuum solution to the Einstein equations is not given by $\mathbb{M}^{4}$, but by de Sitter space. Thus, as we move further away from a localized energy density, one should conjecture that the corresponding spacetime should approach a de Sitter background solution. In this paper, taking this point of view, we aim to present a first candidate for energy in this expanding setting and prove its positivity in some scenarios.

Let us highlight that, in the traditional viewpoint concerning asymptotically Mikowskian systems as models for an isolated system, asymptotic symmetries of $\R^{3,1}$ provide natural tools to propose a notion of energy attached to such a system. This is because, although generic spacetimes possess no symmetries, 
the timelike Killing fields in the asymptotic background of $\R^{3,1}$ can be used to produce such quantities through well-known methods \cite{ChoquetBruhat1984PositiveEnergy, Chrusciel2012, WaldZoupas2000ConservedQuantities}. In the context of de Sitter space, there are no global timelike Killing fields, producing a challenge if one intendeds to follow a similar approach and identify a conserved quantity canonically related to a time-translational background symmetry. Moreover, it is not clear that within an expanding universe with fixed cosmological constant one should expect such a conservation law: On the one hand, matter should get increasingly diluted while on the other hand any contribution from $\Lambda$ remains constant as a density while space expands. These and other challenges motivate us to propose instead a \emph{quasi-local} notion of energy, attaching such a notion to bounded subsets within the cosmological horizon of de Sitter space, instead of a global one, which may be typically guided by such conservation principles. We take this to be a promising first avenue in a program related to better defining isolated systems and their energy in expanding cosmological spacetimes.

Given the above comments, we point out that it might be an interesting related question to understand if a natural global notion of energy could be introduced by considering the background global \emph{conformal} timelike Killing fields of de Sitter space. In particular it is of interest to investigate whether conformal Killing fields could introduce monotonicity laws addressing the subtleties outlined above. Pursuing this alternative is outside the scope of our current work, and inspiration could be drawn from \cite{ChoquetBruhat1984PositiveEnergy, ALM}. If such a program could be successfully carried out, it would then be interesting to understand its relation to our quasi-local definition we put forth.

The structure of the paper is as follows. In Section \ref{sect: energy def} we define our quasi-local notion of energy with an expanding flat de Sitter background, see Definition \ref{def: E_lambda}.
In Section \ref{sect: main result} we state the main conjecture and our main results, Theorem \ref{thm: main} and Corollary \ref{cor: riem}.
Section \ref{sect:quasi-local-results} reviews the relevant notions of quasi-local energy developed in a Minkowskian background, together with the associated positivity results.
In Section \ref{sect: proofs} we provide the relevant lemmas and prove our main results.
Section \ref{sect: example} explores examples of our results. Section \ref{subsect: SchwDeS} reviews the Schwarzschild--de Sitter space in flat expanding coordinates and computes the quantities appearing in Theorem \ref{thm: main}. Section \ref{subsect: SchoenEx} uses a blow-up argument of Schoen to construct many umbilic initial data sets  satisfying the $\Lambda$-vacuum constraint equations \eqref{eq: constraint} where our results apply.
Finally, in Section \ref{sect: conclusions} we summarize our findings and discuss future directions.

\section{Defining energy with an expanding flat de Sitter background}\label{sect: energy def}

An initial data set is a triple $(M,g,k)$, where $(M,g)$ is a Riemannian 3-manifold and $k$ is a symmetric $(0,2)$ tensor on $M$. Two-sided spacelike hypersurfaces in a spacetime $(\mathscr{M}, \mathscr{g})$ are examples of initial data sets; in this case $k$ is the second fundamental form of $M$ within $\mathscr{M}$. In particular, if $u$ is the future-directed unit timelike vector field, then, in our convention,
\begin{equation}\label{eq: defi k}
    k(X,Y) = \langle \nabla_{X}u, Y \rangle,
\end{equation}
where $X, Y$ are tangent vector fields on $M$ and $\nabla$ is the spacetime Levi-Civita connection. With this convention, positive mean curvature of the initial data set corresponds to positive expansion towards the future.


Given $\Lambda \in \R$ (the cosmological constant), we define the \emph{energy density} $\rho$ and momentum density $J$ via,
\begin{equation}\label{eq: constraint}
\begin{aligned}
\rho + \Lambda &= \frac{1}{2}\bigg(R_g + (\tr_g k)^2 - |k|_g^2\bigg), \\
J &= \dv_g \left[ k - (\tr_g k) g \right],
\end{aligned}
\end{equation}
where $ R_g $ is the scalar curvature of the induced metric $ g $. We refer to these equations as the constraint equations; they arise from the Gauss--Codazzi decomposition of the Einstein field equations. 

An initial data set satisfies the \textit{dominant energy condition with respect to $\Lambda$} if the inequality 
 \begin{equation}
     \rho \geq |J|_g 
 \end{equation}
holds. We remark that the dominant energy condition used here is derived from the spacetime dominant energy condition applied to the stress-energy tensor 
$\mathscr{T} = \mathcal{G} + \Lambda \mathscr{g},$ where $\mathcal{G}$ is the Einstein tensor associated with the spacetime metric $\mathscr{g}$. 
Specifically, we require that for every future-directed timelike vector 
$\mathscr{v}$, the vector $-\mathscr{T}(\mathscr{v}, \cdot)^{\sharp}$
is also future-directed timelike. Consequently, this dominant energy 
condition depends explicitly on $\Lambda$.
This condition ensures that energy density dominates momentum flux, reflecting the physical requirement that energy-momentum flows do not exceed the speed of light.

If $\Lambda \geq 0$ and the dominant energy condition with $\Lambda$ holds, then the \emph{geometrical energy density} $\mu$ satisfies
\begin{equation}\label{eq: mu def}
\mu := \rho + \Lambda \geq |J|_g.    
\end{equation}
Note that this is the dominant energy condition without $\Lambda$. Moreover, $\mu$ is determined entirely by the geometry of the initial data:
\begin{equation}\label{eq: mu def 2}
\mu = \frac{1}{2}\bigg(R_g + (\tr_g k)^2 - |k|_g^2\bigg).
\end{equation}

Noether’s theorem (see \cite{arnold1989}), loosely speaking, states that conserved quantities correspond to symmetries of the system. Therefore, to define conserved quantities such as energy in general relativity, we look for global symmetries of the spacetime. These
symmetries are generated by Killing vector fields, which are vector fields that preserve the metric along their flow (i.e. they generate isometries of the spacetime). In particular, a timelike Killing vector field corresponds to a time-translation symmetry and, by Noether’s theorem, gives rise to a conserved energy.
However, general spacetimes do not possess such global symmetries.
This is why asymptotically Minkowskian spacetimes are of particular interest: in these spacetimes the geometry approaches Minkowski space at infinity, where a timelike Killing vector field exists. This asymptotic symmetry provides the structure needed to define a conserved energy at infinity.

The ADM energy \cite{ADM1961} $ E_{\mathrm{ADM}} $ plays a central role in the analysis of asymptotically flat spacetimes. It is derived from the Hamiltonian formulation of general relativity under the assumption of a fixed Minkowski space background. It provides a notion of total energy for an isolated gravitational system and is invariant under asymptotically flat coordinate transformations \cite{bartnik86, Chrusciel1988}. This invariance guarantees that the energy is well-defined independently of the observer’s coordinate choice at infinity.

The ADM energy is defined for initial data sets $ (M, g, k) $ that satisfy specific fall-off conditions at spatial infinity, namely:
\begin{equation}
g_{ij} = \delta_{ij} + O_2(|x|^{-\tau}), \qquad k_{ij} = O_1(|x|^{-1-\tau}),    
\end{equation}
where $\tau > \frac12$. These asymptotic conditions ensure that the metric $ g $ approaches the flat Euclidean metric $\delta$ and that the extrinsic curvature $ k $ decays appropriately. Assuming integrability of $\mu$ and $J$, the energy is defined as
\begin{equation}
    E_{\mathrm{ADM}}= \frac{1}{16 \pi} \lim_{r \to \infty} \int_{S_r} \sum_{i,j} \left(\partial_i  g_{ij} - \partial_j  g_{ii}\right) \frac{x_j}{r} \, d\mu_\delta,
\end{equation}
where $d\mu_\delta$ denotes the area measure on the coordinate sphere $S_r=\{|\mathbf{x}| = r\}$ induced by the Euclidean metric $\delta$. We remark that in this setting the linear momentum and ADM mass are also be well defined.

A fundamental result in mathematical relativity is the \textit{Positive Energy Theorem}. It was first proved for the time symmetric case $k\equiv 0$ independently by Schoen and Yau \cite{SY79} and Witten \cite{witten1981new}. It was later extended for the case of a general $k$ with appropriate asymptotic decay by Schoen and Yau \cite{SY81a,SY81b}. The theorem states that, under the assumption that the dominant energy condition holds, the ADM energy is non-negative:
\begin{equation}
E_{\mathrm{ADM}} \geq 0.
\end{equation}
Equality occurs if and only if the initial data set $ (M, g, k) $ arises from a spacelike slice of flat Minkowski spacetime.

Extensions of the positive energy theorem exist for other asymptotic geometries. In the case $\Lambda < 0$, corresponding to asymptotically hyperbolic manifolds, the notion of energy was introduced by Wang \cite{Wang01} and further developed by Chruściel and Herzlich \cite{ChruscielHerzlich03}. The Riemannian positive mass theorem in this setting was proven under the assumption that the manifold is spin in \cite{Wang01, ChruscielHerzlich03}. The positivity of the energy in the hyperboloidal, non-spin case has been established by Sakovich (see \cite{sakovich2021} and references therein).

When $\Lambda>0$, the situation becomes more delicate. A Positive Energy Theorem in this setting has been established by Borghini and Mazzieri \cite{BorghiniMazzieri2018, BorghiniMazzieri2020} for compact static metrics. See also \cite{AgostinianiBorghiniMazzieri2026} for recent developments on mass-type invariants in the presence of a positive cosmological constant. We also note the connection with rigidity and the Min--Oo conjecture (see \cite{BrendleMarquesNeves}, where it was shown to be false). In contrast with \cite{BorghiniMazzieri2018, BorghiniMazzieri2020, AgostinianiBorghiniMazzieri2026}, our focus is on initial data sets which expand, and hence are not static.

To study isolated systems in this framework, we adopt a background spacetime model that reflects the observed expansion of the universe. The \textit{flat expanding de Sitter model} provides a natural choice, as it captures the large-scale behavior of a universe dominated by $\Lambda$-vacuum energy while retaining the mathematical simplicity of spatial flatness. Our goal is to find a definition of energy in this setting and establish a corresponding positivity result.

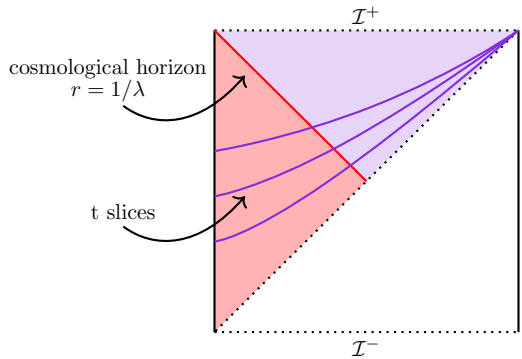
\begin{figure}
    \centering

    \begin{tikzpicture}[scale=1]
\colorlet{inner}{bello!20!white}
\def\a{2}

\coordinate (tl) at (-\a,\a); 
\coordinate (tr) at (\a,\a); 
\coordinate (br) at (\a,-\a); 
\coordinate (bl) at (-\a,-\a);


\fill[color = inner] (tl) -- (tr) -- (bl) --  cycle;
\fill[color = red!30!white] (tl) -- (0,0) -- (bl) --  cycle;
\draw[-,thick]  (tr) -- (br);
\draw[-,thick]  (tl) -- (bl);
\draw[dotted, thick]  (tl) -- (tr);
\draw[dotted,thick]  (bl) -- (br);

\draw[-,thick, red] (0,0) -- (tl) ;
\draw[-,thick, dotted] (tr) -- (bl) ;
\draw[->,thick] (-\a-1.2,1) .. controls (-\a-.8,.7) and (-\a-.2,.7) .. (-\a+.4,1.4);
\node [scale=.75] at (-\a-1.4,1.5)  {cosmological horizon}; 
\node [scale=.75] at (-\a-1.4,1.2)  {$r=1/\lambda$};

\draw [thick,bello] plot [smooth,tension=1] coordinates {(tr) (-.5,0) (-\a,-.8)};
\draw [thick,bello] plot [smooth,tension=1] coordinates {(tr) (-.3,.5) (-\a,-.2)};
\draw [thick,bello] plot [smooth,tension=1] coordinates {(tr) (.1,1) (-\a,0.4)};
\draw[->,thick] (-\a-1.2,-.6) .. controls (-\a-.8,-.9) and (-\a-.2,-.9) .. (-\a+.4,-.2);
\node [scale=.75] at (-\a-1.2,-.4)  {t slices};

\node [scale=.75] at (0,2.2)  {$\mathcal{I}^+$}; 
\node [scale=.75] at (0,-2.2)  {$\mathcal{I}^-$};

\end{tikzpicture}
    
    \caption{One-half of de Sitter space is covered by the spatially flat FLRW model with an exponentially growing scale factor. In flat-expanding coordinates, each $t$-slice is isometric to the Euclidean space $\R^3$ and $\partial_t$ is Killing but timelike only for values within the cosmological horizon, i.e. for $r < 1/\lambda$. Compare with \cite[page 127]{HE}. For an in-depth analysis of how the flat-expanding coordinates relate to the global coordinates of de Sitter space, see \cite{Geshnizjani2023OnTI}.}
    \label{fig:penrose-diagram-deSitter}
\end{figure}

Recall that de Sitter space is the maximally symmetric solution to Einstein's equations with a positive cosmological constant $ \Lambda > 0 $. It serves as a model for a universe undergoing accelerated expansion, which is supported by current observations. In \emph{comoving coordinates}, the metric for the de Sitter space takes the form 
\begin{equation}\label{eq: FLRW metric}
  \mathscr{M}_0 = \R \times \R^3 \quad \text{ and } \quad \mathscr{g}_0 = -d\tau^2 + e^{2\lambda \tau}(d\rho^2 + \rho^2 d\Omega^2),
\end{equation}
where\footnote{The relation $\Lambda = 3\lambda^2$ will hold throughout the rest of the paper.} 
\begin{equation}
\lambda := \sqrt{\frac{\Lambda}{3}}.
\end{equation}
Equation \eqref{eq: FLRW metric} is the standard spatially flat  FLRW model with scale factor $ a(\tau) = e^{\lambda \tau}$.

In \emph{flat-expanding coordinates}, $ t = \tau $, $ r = e^{\lambda \tau}\rho $, the metric takes the form:
\begin{equation}\label{eq: flat expanding de sitter}
    \mathscr{g}_0 = -\left(1 - r^2\lambda^2\right)dt^2 - 2 \lambda r \, drdt + dr^2 + r^2 d\Omega^2.    
\end{equation}
Constant $ t $ slices of the above metric are spatially flat and umbilic, with second fundamental form
\begin{equation}\label{eq: umbilic}
    k_{0} = \lambda \delta.    
\end{equation}
Therefore, this particular foliation of de Sitter space yields an ideal background to model spatially flat systems which undergo cosmological expansion.

In flat-expanding coordinates, the vector field $ \partial_t $ is a Killing vector field; however, it ceases to be timelike for $r$-values greater than the \emph{cosmological horizon}, $r = 1/\lambda$. This leads to complications with an asymptotic definition of energy. The presence of this cosmological horizon suggests the need for a shift from global to \emph{quasi-local} definitions of energy.

Quasi-local notions of energy play a central role in mathematical relativity, as they provide a framework for defining and analyzing the energy content of finite spacetime regions (see \cite{lee2019} for a survey of various notions of quasi-local energy and their properties).

In the quasi-local setting, one considers an initial data set $(\Omega, g, k)$, where $\Omega$ is compact with with boundary\footnote{In general, $\Omega$ can have multiple boundary components.} $\S := \partial \Omega$ of positive Gauss curvature. By the Weyl embedding theorem \cite{Nirenberg1953, Pogorelov1952}, there is an isometric embedding of $\Sigma$ within Euclidean space $\R^3$ which bounds a compact and convex subset $\Omega_0$. Let $H_0$ and $H$ denote the mean curvatures of $\Sigma$ within $\R^3$ and $\Omega$, respectively. (In our convention, the mean curvature is the trace of the second fundamental form with respect to the outward unit normal.) Then we define the energy $E_\lambda$ as a direct generalization of the Liu--Yau energy \cite{LiuYau2003, LiuYau2006}. It can be derived from the free-energy Hamiltonian in general relativity, see eq. (105) in \cite{Kijowski1997}. 
\
\begin{definition}\label{def: E_lambda}
Let $(\Omega, g, k)$ be a compact initial data set with nonempty boundary $\Sigma := \partial \Omega$ of positive Gauss curvature. Let $H
$ be the mean curvature of $\Sigma$ in $\Omega$ with respect to the outward normal, and let $H_0$ be the mean curvature of the isometric embedding of $\Sigma$ into Euclidean space $\R^3$. Then
\begin{equation}\label{eq: E_lambda def}
    E_\lambda := \frac{1}{8\pi} \int_\Sigma \left( \sqrt{H_0^2 - 4\lambda^2} \,-\, \sqrt{H^2 - (\tr_\Sigma k)^2} \right) d\mu.
\end{equation}
\end{definition}

As with the Liu--Yau energy, $E_\lambda$ can be expressed as the difference between the norms of the codimension-two mean curvature vectors in the relevant spacetimes. Indeed it's not hard to see that 
\begin{equation}\label{eq: defi mcv}
E_\lambda = \frac{1}{8\pi} \int_\Sigma \big(|\vec{H}_0| \,- \,|\vec{H}|\big)\, d\mu.
\end{equation}
Here $\vec{H}$ is the mean curvature vector of $\Sigma$ (assumed to be spacelike) within the original spacetime $(\mathscr{M},\mathscr{g})$ containing $(\Omega, g,k)$, while $\vec{H}_0$ is the mean curvature vector of $\Sigma$ within the de Sitter space $(\mathscr{M}_0,\mathscr{g}_0)$ with $\Lambda = 3\lambda^2$ containing the compact initial data set $(\Omega_0, \delta, k_0 = \lambda \delta)$ within a constant $t$-slice. See figure \ref{fig: penrose of LY} for a comparison between the Liu-Yau energy and $E_\lambda$.  Like the Liu--Yau energy, inherent in the definition of $E_\lambda$, is the specific choice of Weyl embedding; different embeddings into de Sitter space will yield different energies. See section \ref{sect: dep on embed}.

The goal of this paper is to establish positivity results for $E_\lambda$. Analogous to the classical Liu--Yau framework, this requires ensuring that the data under consideration can be embedded isometrically into $\mathbb{R}^3$. Moreover, for the definition to be physically meaningful in a de Sitter background, the embedding must lie within the cosmological horizon $r< 1/\lambda$ where the Killing field $\partial_t$ is timelike (see figure \ref{fig: penrose of LY}).

In Section \ref{sect: main result} we first state a general conjecture concerning the positivity of $E_\lambda$. The conjecture naturally leads to a two-fold problem: 
\begin{enumerate}
    \item a purely geometric question, concerning the conditions under which a given surface $\Sigma$ can be isometrically embedded inside the cosmological horizon;
    \item a question regarding the positivity of the energy formula given the geometric data.
\end{enumerate}
In this work, we address both issues and provide partial answers for small values of $\lambda.$ 

We remark that the observed value for the cosmological constant is, in fact, small. Using the observations made by \cite{Planck_2020}, we have $\Lambda_{\rm obs} = 1.10\times 10^{-52}$ ${\rm m}^{-2}$. For this value, we have $\lambda_{\rm obs} = 6.05 \times 10^{-27}$ ${\rm m}^{-1}$. The cosmological horizon at this observed value has a radius of $r = 1/\lambda_{\rm obs} = 1.43 \times 10^{26}$ ${\rm m}$, which is about 100,000 times greater than the diameter of the Milky Way.

\begin{figure}   
\begin{center}
\begin{tikzpicture}[scale=1]
\colorlet{inner}{bello!20!white}

\begin{scope}[shift={(-1.5,0)}]
\def\a{2}

\coordinate (tl) at (-\a,\a); 
\coordinate (tr) at (\a,\a); 
\coordinate (br) at (\a,-\a); 
\coordinate (bl) at (-\a,-\a);

\fill[color = inner] (tl) -- (0,0) -- (bl) --  cycle;
\draw[dotted, thick]  (tl) -- (0,0);
\draw[dotted,thick]  (bl) -- (0,0);
\draw[thick] (tl) -- (bl);

\draw [thick,bello] (-\a,0) -- (0,0);

\node [scale = .5] [circle, draw, fill = black] at (-1,0)  {};
\node [scale = .75] at (-1,-.25)  {$\S$};

\node [scale=.75] at (-.7,1.1)  {$\mathcal{I}^+$}; 
\node [scale=.75] at (-.7,-1.1)  {$\mathcal{I}^-$};

\end{scope}

\begin{scope}[shift={(3,0)}]
\def\a{2}

\coordinate (tl) at (-\a,\a); 
\coordinate (tr) at (\a,\a); 
\coordinate (br) at (\a,-\a); 
\coordinate (bl) at (-\a,-\a);
\fill[color = inner] (tl) -- (tr) -- (bl) --  cycle;
\draw[-,thick]  (tr) -- (br);
\draw[-,thick]  (tl) -- (bl);
\draw[dotted, thick]  (tl) -- (tr);
\draw[dotted,thick]  (bl) -- (br);

\draw[-,thick, red] (0,0) -- (tl) ;
\draw[-,thick, dotted] (tr) -- (bl) ;

\draw [thick,bello] plot [smooth,tension=1] coordinates {(tr) (-.3,.5) (-\a,-.2)};

\node [scale = .5] [circle, draw, fill = black] at (-1,.15)  {};
\node [scale = .75] at (-1,-.1)  {$\S$};

\node [scale=.75] at (0,2.2)  {$\mathcal{I}^+$}; 
\node [scale=.75] at (0,-2.2)  {$\mathcal{I}^-$};

\end{scope}

\end{tikzpicture}
\end{center}
\caption{Left: the Penrose diagram of Minkowski spacetime. To compute the Liu--Yau energy, the surface $\Sigma$ is isometrically embedded in a constant $t$-slice of Minkowski spacetime. Right: the Penrose diagram of de Sitter space. The surface is also isometrically embedded in a constant $t$-slice of the flat de Sitter space and lies within the cosmological horizon (shown in red).}
 \label{fig: penrose of LY}   
\end{figure}
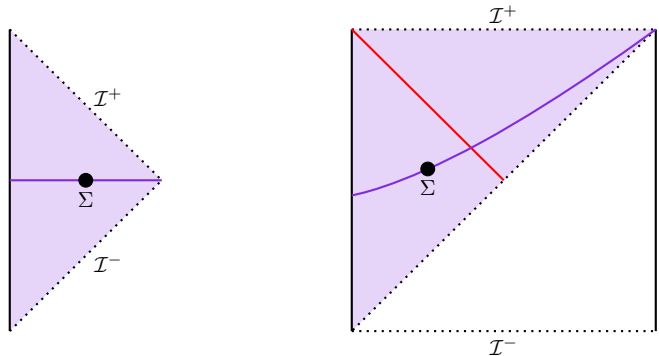

\section{Main conjecture and theorem}\label{sect: main result}

As explained above, the presence of a cosmological horizon in the flat-expanding coordinates of de Sitter space suggests the need to replace the usual notion of total energy at infinity with a quasi-local one. Motivated by this, we introduce the energy $E_\lambda$ defined in \eqref{eq: E_lambda def}, which is a direct generalization of the Liu--Yau quasi-local energy \cite{LiuYau2003, LiuYau2006} to the flat-expanding patch of de Sitter space. The relevant definitions and properties of the classical Liu--Yau construction will be recalled in Section \ref{sect:quasi-local-results}.

Our goal is to establish positivity results for $E_\lambda$. In analogy with the classical Liu--Yau framework, this requires verifying that the data under consideration can be embedded isometrically into $\mathbb{R}^3$. Furthermore, for the definition to be physically meaningful in the de Sitter case, the embedding must lie inside the cosmological horizon of the background spacetime (see Figure \ref{fig: penrose of LY}).

Consider an initial data set $(\Omega,g,k)$ where $\Omega$ is compact with non-empty boundary $\S := \partial \Omega$ of positive Gauss curvature.  By the Weyl embedding theorem, which was solved independently by Nirenberg and Pogorelov \cite{Nirenberg1953,Pogorelov1952}, $\Sigma$ can be isometrically embedded into the Euclidean space $\mathbb{R}^3$ with the image of $\Sigma$ as the boundary of a compact and convex domain $\Omega_0 \subset \R^3$; moreover, this embedding is unique up to Euclidean rigid motions. Let $H_0$ be the mean curvature of the embedding in $\R^3$ with respect to the outward normal. In the initial data $\Omega$, let $H$ be the mean curvature of $\Sigma$ with respect to outward normal. 


Our first conjecture states that the energy $E_\lambda$ should be nonnegative whenever the initial data set satisfies the dominant energy condition with respect to $\Lambda = 3\lambda^2$ and the Weyl embedding of $\Sigma$ in $\R^3$ lies within the cosmological horizon; as in Liu-Yau \cite{LiuYau2003, LiuYau2006}, we also require that the mean curvature vectors of $\Sigma$ are spacelike within $\Omega$ and the Weyl embedding in $\R^3$ so that $E_\lambda$ is well defined.

\begin{conj}\label{conj: main}
Fix $\l \geq 0$. Let $(\Omega, g, k)$ be a compact initial data set with nonempty boundary $\Sigma := \partial\Omega$ of positive Gauss curvature $K > 0$. Assume that:
\begin{enumerate}
  \item the Weyl embedding of $\Sigma$ in $\mathbb{R}^3$ lies entirely inside the cosmological horizon, $r < 1/\lambda$;
  \item $H_0 > 2\lambda$;
  \item $H > |\tr_\Sigma k|$;
  \item the initial data set $(\Omega,g,k,\lambda)$ satisfies the dominant energy condition with $\Lambda = 3\lambda^2$.
\end{enumerate}
Then the quasi-local energy
\begin{equation}\label{eq: energy_in_conj1}
    E_\lambda := \frac{1}{8\pi} \int_\Sigma \left( \sqrt{H_0^2 - 4\lambda^2} \,-\, \sqrt{H^2 - (\tr_\Sigma k)^2} \right) d\mu
\end{equation}
is non-negative:
\begin{equation}
    E_\lambda \ge 0.
\end{equation}
\end{conj}

In the umbilic setting, $k = \lambda g$, the dominant energy condition in question is a purely Riemannian condition. Indeed, from the first constraint equation in \eqref{eq: constraint}, the dominant energy condition with $\Lambda = 3\lambda^2$ is equivalent to requiring nonnegativity of scalar curvature, i.e. $R_g \geq 0$. This leads to the following purely Riemannian version of Conjecture \ref{conj: main}.

\begin{conj}\label{conj: riem}
Fix $\lambda \geq 0$. Let $(\Omega, g, k = \lambda g)$ be a compact initial data set with nonempty boundary $\Sigma := \partial\Omega$ of positive Gauss curvature $K > 0$. Assume that:
\begin{enumerate}
  \item the Weyl embedding of $\Sigma$ in $\mathbb{R}^3$ lies entirely inside the cosmological horizon, $r < 1/\lambda$;
  \item $H_0 > 2\lambda$;
  \item $H > 2\lambda$;
  \item $R_g \geq 0$.
\end{enumerate}
Then the quasi-local energy
\begin{equation}\label{eq: energy_in_conj2}
    E_\lambda = \frac{1}{8\pi} \int_\Sigma \left( \sqrt{H_0^2 - 4\lambda^2} \,-\, \sqrt{H^2 - 4\lambda^2} \right) d\mu
\end{equation}
is non-negative:
\begin{equation}
    E_\lambda \ge 0.
\end{equation}
\end{conj}

Addressing these conjectures involves two challenges. The first is geometric in nature: one would like to determine under which conditions the surface $\S$ can be isometrically embedded inside the cosmological horizon. The second concerns the positivity of the energy $E_\l$ itself. We will show that both of these concerns can be resolved for some values of $\lambda$.
\\

We start by introducing a constant $\alpha$ which depends on the initial data $(\Omega, g, k)$ via
 \begin{equation}\label{eq:def-lambdamax}
    \a = 
        \begin{cases}
            \frac{2\sqrt6}{3\pi} \sqrt{K_{\min}} & \text{if }  \sqrt{K_{\min}} < \frac{E_{\LY}}{r_\S^2}\\
            \min \left\lbrace \sqrt{\frac{E_{\LY}}{r_\S^2} \left( 2\sqrt{K_{\min}}-\frac{E_{\LY}}{r_\S^2} \right)} , \frac{2\sqrt6}{3\pi} \sqrt{K_{\min}}  \right\rbrace & \text{if }  \sqrt{K_{\min}} \geq \frac{E_{\LY}}{r_\S^2}
        \end{cases}.
    \end{equation}
In this expression, $K_{\min}$ is the minimum of the Gauss curvature $K$ on $\Sigma$, $r_\Sigma$ is the area radius of $\Sigma$, and $E_{\LY}$ is the Liu-Yau energy of the initial data (see equation \eqref{eq: LY energy 1} in the next section). We note that the definition of $\alpha$ does not involve $\lambda$ explicitly; however, an implicit dependence may occur if either the metric $g$ or the second fundamental form $k$ (and hence the Liu--Yau energy $E_{\LY}$) depends on $\lambda$.

Our main result, Theorem \ref{thm: main}, addresses Conjecture \ref{conj: main} for some values of $\lambda$.  
Then we will find that $E_\l \geq 0$ so long as $\lambda\leq \alpha_\lambda $. Recognize that $\alpha_\lambda$ is nonnegative whenever the initial data set satisfies the dominant energy condition with respect to $\Lambda \geq 0$. Moreover $\alpha_\lambda$ is strictly positive if the geometrical energy density $\mu$ is strictly positive, which follows from the rigidity statement of Theorem \ref{thm: LY}.

\begin{theorem}\label{thm: main} 
Fix $\lambda \geq 0$. Let $  (\Omega, g, k) $ be a compact initial data set with nonempty boundary $\Sigma := \partial \Omega$ of positive Gauss curvature $K > 0$ and $H > |\tr_\Sigma k|$. Assume that the dominant energy condition with respect to $\Lambda = 3\lambda^2$ holds. If $\lambda \leq \alpha$, where $\alpha$ is given by \eqref{eq:def-lambdamax}, then the surface $\Sigma$ isometrically embeds inside the cosmological horizon $r < 1/\lambda$ in $\R^3$, and the energy \eqref{eq: energy_in_conj1} is nonnegative:
\begin{equation}
     E_\lambda \geq 0. 
\end{equation} 
Moreover, if $0 < \l < \a$, then $E_\l >0$; if $\l = \a >0$ and $E_{\lambda} = 0$, then $\Sigma$ is isometric to a disjoint union of round spheres.
\end{theorem}

%


\begin{remark}
    We note that Theorem \ref{thm: main} holds for any choice of the second fundamental form $k$. In particular, we can adopt two different points of view:
    \begin{enumerate}
        \item \textbf{Fixed geometric data.} One may fix an initial data set $(\Omega, g, k)$ that is completely independent of $\lambda$. 
        In this case, $\lambda$ acts as an external parameter, and $\alpha$ is independent of $\lambda$. Therefore, in this case, the theorem provides an interval $[0, \alpha]$  of admissible values of $\lambda$ for which the data embeds within the cosmological horizon and the quasi-local energy $E_\lambda$ is nonnegative.  
        \item \textbf{$\lambda$-dependent data.} Alternatively, one may consider a $\lambda$-dependent family of data $(\Omega, g(\lambda), k(\lambda))$, for instance determined by the constraint equations or by prescribing an umbilic second fundamental form. 
        In this case, $\alpha$ depends on $\lambda$ since $E_{\LY}$ does. The theorem can be applied pointwise by verifying the condition $\lambda \leq \a$ for each value of $\lambda$.
    \end{enumerate}
\end{remark}

Concerning the upper bound for $\lambda$ in Theorem \ref{thm: main} above, in our proof this appears as a technical restriction, which we believe to be non-optimal. This arises due to our appeals to Lemmas \ref{lem:embedding-in-cosmo} and \ref{lem: positivity}. Specifically, we appeal to Lemma \ref{lem:embedding-in-cosmo} in order to grant the the Weyl embedding $\mathcal{W}:\Sigma\to \mathbb{R}^3$ isometrically embeds inside the cosmological horizon and for this we use a non-optimal estimate appealing to \cite{Jung1901} in order to produce a bound on the diameter of the enclosed region by $\mathcal{W}(\Sigma)$, which is suboptimal. This bound is what in the end sets an upper bound for $\lambda$ in the proof of Lemma \ref{lem:embedding-in-cosmo}. Similarly, we appeal to Lemma \ref{lem: positivity} to establish positive, where this lemma appeal to the stated upper bound for $\lambda$ to be able to reduce the positivity statement to that of the Liu--Yau energy $E_{\rm LY}$ (see Definition \ref{DefnLYEnergy}). Let us nevertheless highlight, that despite the introduction of such upper bounds for $\lambda$, Theorem \ref{thm: main} provides a geometrically determined such bound, establishing that for values of the cosmological constant which are not too large in comparison with clear-cut geometric quantities, the statement holds. Moreover, we emphasize that the actually observed value for the cosmological constant in our universe is remarkably small, and thus that, although maybe mathematically suboptimal, positive energy theorems for small $\Lambda$ are physically meaningful.

We now address Conjecture \ref{conj: riem}, the Riemannian version of Conjecture \ref{conj: main}. In particular, we show that there is a given interval of values of $\lambda>0$ for which the data correctly embeds within the cosmological horizon and has nonnegative energy $E_\lambda$. This interval will be only dependent on the geometric data $(\Omega,g)$.

\begin{corollary}\label{cor: riem} 
Fix $\lambda \geq 0$. Let $(\Omega, g)$ be a compact Riemannian manifold  of nonnegative scalar curvature with nonempty boundary $\Sigma := \partial \Omega$ of positive Gauss curvature $K > 0$. Then there exists a constant $\lambda_{\max} \geq 0$ only dependent on $(\Omega,g)$ such that, for any $\lambda \in [0, \lambda_{\rm max}]$, the surface $\Sigma$ isometrically embeds inside the cosmological horizon $r < 1/\lambda$ in $\mathbb{R}^3$, and if the mean curvature of $\Sigma$ in $(\Omega, g)$ satisfies $H > 2\lambda$, then the energy \eqref{eq: energy_in_conj2} of the initial data set $(\Omega, g, k = \lambda g)$ is well defined and nonnegative:
\begin{equation}
    E_\lambda \geq 0.
\end{equation}
Moreover, it's strictly positive if $0 < \lambda \leq \lambda_{\rm max}$.
\end{corollary}

\begin{remark}
    The constant $\lambda_{\max}$ in Corollary \ref{cor: riem} is defined as $\a$ in \eqref{eq:def-lambdamax} by substituting $E_{\LY}$ with its Brown--York counterpart $E_{\BY}$ everywhere. Consequently, by the rigidity statement of Shi--Tam \cite{ShiTam2002} (see Theorem \ref{thm: ST} below), $\lambda_{\rm max} > 0$ whenever $(\Omega,g)$ is not a subset of Euclidean space $\R^3$. For specific examples of this corollary, see Section \ref{subsect: SchoenEx}.
\end{remark}


\section{Positivity Results for Quasi-Local Energies with a Minkowski background}\label{sect:quasi-local-results}

In this section we provide some motivation for the quasi-local energy $E_\lambda$ defined by \eqref{eq: E_lambda def}. Briefly, the expression for $E_\lambda$ is a generalization of the Brown--York and Liu--Yau definitions of energy but adapted to the setting of an expanding flat de Sitter background instead of a Minkowski background.

We begin by recalling the classical result concerning isometric embeddings of positively curved 2-surfaces. Let $\Sigma$ be a closed surface with a Riemannian metric of positive Gauss curvature. Then there exists an isometric embedding of $\Sigma$ into $\R^3$, known as the Weyl embedding, which bounds a compact convex region $\Omega_0 \subset \R^3$. Moreover, the embedding is unique up to Euclidean rigid motions \cite{Nirenberg1953,Pogorelov1952}. In particular, the mean curvature of the isometric embedding is uniquely determined by the intrinsic metric. Using this fact, Brown and York proposed the following quasi-local mass:

\begin{definition}[Brown--York Energy]\label{def: BY energy}
Let $(\Omega,g)$ be a compact Riemannian 3-manifold with nonempty boundary $\Sigma := \partial\Omega$ of positive Gauss curvature. Let $H$ be the mean curvature of $\Sigma$ in $\Omega$ with respect to the outward normal, and let $H_0$ the mean curvature of the isometric embedding of $\Sigma$ into $\mathbb{R}^3$. Then the Brown--York energy of $\Sigma$ is defined as
\begin{equation}\label{eq: Brown-York energy}
E_{\BY} := \frac{1}{8\pi}  \int_\Sigma \left(H_0\,  -  H  \right) d\mu.
\end{equation}
\end{definition}
The positivity of the Brown--York energy was established by Shi and Tam \cite{ShiTam2002} assuming the dominant energy condition and mean convexity.

\begin{theorem}[\cite{ShiTam2002}]\label{thm: ST}
Suppose $(\Omega,g)$ is a compact Riemannian 3-manifold with non-empty boundary $\Sigma := \partial \Omega$ of positive Gauss and mean curvatures. If $(\Omega,g)$ has nonnegative scalar curvature, then the Brown--York energy $E_{\BY}$ of $\Sigma$ is nonnegative and equals zero if and only if $\Omega$ is isometric to a domain in the Euclidean space $\R^3$ and $\Sigma$ is connected.
\end{theorem}

Liu and Yau \cite{LiuYau2003, LiuYau2006} (see also Kijowski \cite{Kijowski1997}) generalized Definition \ref{def: BY energy} to initial data sets.

\begin{definition}\label{DefnLYEnergy}
Let $(\Omega, g, k)$ be a compact initial data set with nonemtpy boundary $\Sigma := \partial \Omega$ of positive Gauss curvature. Suppose the mean curvature of $\Sigma$ within $\Omega$ satisfies $H > |\text{tr}_\Sigma k|$. Let $H_0$ be the mean curvature of the isometric embedding of $\Sigma$ into $\R^3$. Then the Liu--Yau energy of $\Sigma$ is defined as 
\begin{equation}\label{eq: LY energy 1}
     E_{\mathrm{LY}} = \frac{1}{8\pi } \int_\Sigma \left(H_0\, - \sqrt{H^2 - (\tr_\S k)^2}\right) d\mu . 
\end{equation}
\end{definition}
If the initial data set $(\Omega, g, k)$ is embedded within a spacetime, then the assumption on the mean curvature appearing Definition \ref{DefnLYEnergy} implies that the mean curvature vector $\vec{H}$ of $\Sigma$ is spacelike. In fact, we can write the Liu--Yau energy as
\begin{equation}\label{eq: LY energy 2}
    E_{\mathrm{LY}} := \frac{1}{8\pi} \int_\Sigma \big(H_0 \,- \,|\vec{H}|\big)\, d\mu .    
\end{equation}
From this expression, we see that the definition of $E_{\rm LY}$ does not depend on the specific slice $\Omega$ of spacetime which $\Sigma$ lies in.




Liu and Yau established positivity under the dominant energy condition (with $\Lambda = 0$).

\begin{theorem}[\cite{LiuYau2003, LiuYau2006}]\label{thm: LY}
Suppose $(\Omega,g,k)$ is a compact initial data set satisfying the dominant energy condition with $\Lambda = 0$ and has nonempty boundary $\Sigma := \partial \Omega$ of positive Gauss curvature and $H > |\tr_\S k|$. Then the Liu--Yau energy $E_{\mathrm{LY}}$ of $\Sigma$ is non-negative and equals zero only if $(\Omega, g, k)$ can be embedded within Minkowski spacetime $\R^{3,1}$ and $\Sigma$ is connected.
\end{theorem}

However, Ó Murchadha, Szabados, and Tod \cite{OMurchadha2004} constructed examples showing that there exist surfaces in Minkowski space satisfying these assumptions for which both the Brown--York and Liu--Yau masses are strictly positive.

To address this issue, Wang and Yau \cite{WangYau1, WangYau2} introduced a new quasi-local energy based on isometric embeddings of 2-surfaces into Minkowski space. 
Their construction involves finding an embedding that minimizes a certain energy functional associated with the surface, and they proved existence and uniqueness results for such embeddings. This leads to a new expression of quasi-local mass for a large class of admissible surfaces.  They showed that such mass is positive whenever the ambient spacetime satisfies the dominant energy condition, vanishes for surfaces in Minkowski space, and satisfies a rigidity property: if the mass is zero, the surface is isometric to one in flat space (see \cite{WangYau2}). 

We generalize the example of \cite{OMurchadha2004} to our de Sitter space setting in section \ref{sect: dep on embed}. This suggests an interesting direction of further research: generalize the Wang--Yau notion of quasi-local energy to the setting of de Sitter space.


\section{Supporting lemmas and proof of main results}\label{sect: proofs}

In this section, we prove some supporting lemmas and Theorems \ref{thm: main} and \ref{cor: riem}. 
The first lemma provides a sufficient condition on the Gauss curvature of $\Sigma$ that ensures it can be isometrically embedded inside the cosmological horizon $r=1/\lambda$ in $\R^3$. 
The second and third lemmas address the positivity of the energy $E_\lambda$. 
Finally, we present the proofs of the main theorems by combining the lemmas.

\subsection{Isometric embedding inside the cosmological horizon}
\: 

\medskip

The next lemma gives a sufficient (albeit not optimal) condition for $\S$ to lie inside the cosmological horizon $r < 1/\lambda$ once embedded in Euclidean space $\R^3$. The remark that appears after Lemma \ref{lem:embedding-in-cosmo} provides an example showing that the condition is not optimal. As the proof shows, the lower bound condition is assumed to appeal to Jung's theorem. It would be interesting if a more precise condition could be formulated.

\begin{lemma}\label{lem:embedding-in-cosmo} 
   Fix $\lambda \geq 0$. Let $\Sigma$ be a Riemannian two-sphere, and suppose the Gauss curvature $K$ of $\Sigma$ is positive and bounded below by
    \begin{equation}\label{eq:K-bound}
        K \geq \frac{3}{8}\pi^2 \l^2.
    \end{equation}
    Then the isometric image of $\Sigma$ under the Weyl embedding into $\R^3$ is contained in the interior of a 3-ball with radius $1/\l$.
\end{lemma}

\begin{proof}

    In two dimensions, the Ricci tensor satisfies $\Ric = K \gamma$, where $\gamma$ is the metric on $\Sigma$. Using \eqref{eq:K-bound} and the Bonnet–Myers theorem, we obtain a bound on the diameter of $\Sigma$,
\begin{equation}
    \mathrm{diam}(\Sigma) \leq \sqrt{\frac{8}{3}}\,\frac{1}{\l}. \label{eq:diam-sigma}
\end{equation}

Let $\Omega_0$ denote the compact and convex region enclosed by $\Sigma$ within $\R^3$ (determined by the Weyl embedding theorem). By convexity, the diameter of $\Omega_0$\footnote{To be precise, $\mathrm{diam}(\Omega_0) := \displaystyle\sup_{x,y \in \Omega_0}|x-y|$ while $\mathrm{diam}(\Sigma) := \displaystyle\sup_{x,y\in \Sigma}d_\Sigma(x,y)$.} is strictly less than that of its boundary:
\begin{equation}
    \mathrm{diam}(\Omega_0) < \mathrm{diam}(\Sigma). \label{eq:diam-omega}
\end{equation}
To see this,  since $\Omega_0$ is compact and convex, its diameter is realized by a straight line in $\Omega_0$. That is, there are points $p,q \in \Omega_0$ such that $\text{diam}(\Omega_0) = |p-q|$ which is the length of the straight line $c \colon [0,1] \to \Omega_0$ given by $c(t) = tp + (1-t)q$. By convexity, $p$ and $q$ lie on the boundary $\Sigma$. Therefore
\begin{equation}\label{eq: < diam Sigma}
|p - q|  < d_\Sigma(p,q) \leq \text{diam}(\Sigma).
\end{equation}
The strict inequality in \eqref{eq: < diam Sigma} follows since $
K> 0$. Indeed, seeking a contradiction, suppose $|p-q| = d_\Sigma(p,q)$. Then the image of $c$ must lie in $\Sigma$ (minimizing curves in Euclidean space are unique); hence $c$ is a pregeodesic in $\Sigma$. But, since $c$ is a geodesic in Euclidean space, it follows that the second fundamental form of $\Sigma$ within $\R^3$ vanishes when evaluated on $c'$; this contradicts the fact that $\Sigma$ has positive Gauss curvature.

Jung’s theorem \cite{Jung1901} shows that $\Omega_0$ is enclosed in a ball of radius 
\begin{equation}
    r = \sqrt{\frac{3}{8}}\,\mathrm{diam}(\Omega_0). \label{eq:jung}
\end{equation}
Combining \eqref{eq:diam-sigma}, \eqref{eq:diam-omega}, and \eqref{eq:jung}, we conclude that $\Sigma$ is contained in a ball of radius $r< \frac{1}{\l}$, which completes the proof.
    \end{proof}

\begin{remark}
    The curvature bound in equation \eqref{eq:K-bound} is not sharp. Indeed, consider a round sphere of radius $r = \frac{1}{\lambda + \varepsilon}$, for $\varepsilon > 0$.
    Such a sphere clearly lies within the cosmological horizon.
    However, its Gauss curvature is $K = (\lambda + \varepsilon)^2,$ which violates the bound in \eqref{eq:K-bound} whenever $\varepsilon$ is sufficiently small (note that $\tfrac{3}{8}\pi^2 \approx 3.7)$. Indeed \eqref{eq:K-bound} is only a sufficient condition for our proof and the Jung's theorem radius would be realized by a tetrahedron.
\end{remark}

\subsection{Positivity of energy}

\:
\medskip

Our goal is to establish the positivity of the energy $E_\lambda$. As in Liu--Yau, we always assume $H > |\text{tr}_\Sigma k|$.  

Let $(\Omega, g, k)$ be a compact initial data set with boundary $\Sigma = \partial \Omega$ of positive Gauss curvature $K > 0$ satisfying the dominant energy condition for $\Lambda = 3 \lambda^2 \geq 0$.

Let $K_{\text{min}} = \min_\S K$. We define 
\begin{equation}\label{eq:def-lambda0}
    \a_0 = \sup_{C\in (0,1]} \min \left\lbrace C \sqrt{K_{\text{min}}}, \frac{ 1+ \sqrt{ 1-C^2}  }{ C } \frac{E_{\LY} }{r_\Sigma^2}\right\rbrace,
\end{equation}
where $r_\Sigma$ is the area radius of $\S$ and $E_{\LY}$ is the Liu--Yau energy of $\Sigma$ given by \eqref{eq: LY energy 1}.

Note that $\a_0 \geq 0$ since $E_{\LY} \geq 0$ (since $\mu \geq |J|_g$). Moreover, if $\mu > 0$, the rigidity part of Theorem \ref{thm: LY} implies that $E_{\LY} > 0$ and hence $\a_0 > 0$; in this case, a little algebra shows that $\a_0$ is realized by a unique $C_0 \in (0,1]$ and calculated below:

   \begin{equation}\label{eq: def-lambda_0}
    \a_0 = 
        \begin{cases}
            \sqrt{K_{\min}} & \text{if }  \sqrt{K_{\min}} < \frac{E_{\LY}}{r_\S^2}\\
            \sqrt{\frac{E_{\LY}}{r_\S^2} \left( 2\sqrt{K_{\min}}-\frac{E_{\LY}}{r_\S^2} \right)} & \text{if }  \sqrt{K_{\min}} \geq \frac{E_{\LY}}{r_\S^2}
        \end{cases}.
    \end{equation}

The next lemma shows that for fixed geometric data, $E_\lambda \geq 0$ provided $\lambda$ is not too large.

\begin{lemma}\label{lem: positivity}
   Fix $\lambda \geq 0$. Let $(\Omega, g, k)$ be a compact initial data set with nonempty boundary $\Sigma := \partial \Omega$ of positive Gauss curvature $K >0$ and $H > |\tr_\S k|$. Assume that the dominant energy condition with respect to $\Lambda = 3\lambda^2 \geq 0$ holds. If $\l \in [0, \a_0]$, where $\a_0$ is defined by \eqref{eq:def-lambda0}, then the energy \eqref{eq: energy_in_conj1} is well defined and nonnegative:
    \begin{equation}
        E_\l \geq 0.
    \end{equation}
Moreover, if $0 < \l < \a_0$, then $E_\l >0$, and if $E_{\a_0} = 0$ with $\a_0 > 0$, then $\Sigma$ is isometric to a round sphere and $\sqrt{K_{\min}} \geq \frac{E_{\rm{LY}}}{r_\Sigma^2}$.
\end{lemma}

\begin{proof}
    For any $C\in (0,1]$, define
    \begin{equation}\label{eq: alpha_C}
    \a_C = \min \left\lbrace C \sqrt{K_{\text{min}}}, \frac{ 1+ \sqrt{ 1-C^2}  }{ C } \frac{E_{\LY} }{r_\Sigma^2}\right\rbrace.
    \end{equation}
Note that $\a_0 = \sup_{C \in (0,1]}\a_C$. Moreover, if $\a_0 > 0$, then the map $(0,1] \mapsto (0, \a_0]$ given by $C \mapsto \a_C$ is surjective.

Let $\tf A_0$ denote the trace-free part of the second fundamental form of $\S$ as embedded in Euclidean space $\R^3$. For any $C \in (0,1]$, the Gauss equation implies 
    \begin{equation}\label{eq: H_0 bound}
        H_0 =\sqrt{ 2 |\tf A_0|^2 + 4 K } \geq 2\sqrt{K} \geq \frac{2\a_C}{C}.
    \end{equation}

If $\lambda = 0$, then $E_{\lambda = 0} = E_{\LY} \geq 0$. On the other hand, if $\lambda \in (0, \a_0]$, choose $C \in (0,1]$ such that $\alpha_C = \lambda$.
 Then \eqref{eq: H_0 bound} shows that $E_\lambda$ is well defined and 
 \begin{equation}\label{eq: E_lambda geq 0}
        \begin{split}
            E_{\l} &= E_{\alpha_C}
            \\
            &= \frac{1}{8\pi} \int_\S \left( \sqrt{H_0^2 - 4\alpha_C^2} - \sqrt{H^2-(\tr_\Sigma k)^2}\right)d\mu\\
            &= \frac{1}{8\pi} \int_\S \left(  H_0 - \sqrt{H^2-(\tr_\Sigma k)^2} + \sqrt{H_0^2 - 4\alpha_C^2}- H_0 \right)d\mu\\
            &= E_{\LY}  -  \frac{1}{8\pi} \int_\S \frac{4\alpha_C^2}{H_0} \frac{1}{1+\sqrt{1-4\alpha_C^2/H_0^2}} \\
            &\geq E_{\LY}  -  \frac{|\S|}{8\pi} \frac{2C\alpha_C}{1+\sqrt{1-C^2}} \\
            &= E_{\LY} -  \frac{C}{1+\sqrt{1-C^2}} \alpha_C r^2_\Sigma \\
            &\geq 0 .
        \end{split}
    \end{equation}
In the second to last inequality, we used \eqref{eq: H_0 bound}; in the last inequality, we used \eqref{eq: alpha_C}.  
    
    Note that if $\a_C \neq \a_0$, then one inequality in \eqref{eq: E_lambda geq 0} becomes strict. On the other hand, suppose $E_{\a_0} = 0$ with $\a_0 > 0$. Then all the inequalities in \eqref{eq: E_lambda geq 0} become equalities with $\alpha_C$ replaced by $\a_0$ and $C$ is replaced by the unique value $C_0$ which realizes $\a_0$. This forces \eqref{eq: H_0 bound} to become an equality in $\a_0$ and $C_0$. Hence $\tf A_0 = 0$ and so the induced metric on $\Sigma$ is round \cite[Prop. 4.36]{ONeill1983}. Lastly, if $\sqrt{K_{\min}} < \frac{E_{\LY}}{r_\Sigma^2}$, then $C_0 = 1$ and the last inequality in \eqref{eq: E_lambda geq 0} becomes strict.
\end{proof}

\subsection{Proof of main results.}
\:

\medskip

Combining the lemmas above, we can now prove Theorem \ref{thm: main}.

\begin{proof}[Proof of Theorem \ref{thm: main}]
By definition, we have $\alpha \leq  \frac{2\sqrt{6}}{3\pi}\sqrt{K_{\rm min}}$ and so $\alpha \leq \alpha_0$, where $\alpha_0$ is given by \eqref{eq: def-lambda_0}. Therefore, since $\lambda \leq \alpha$ by hypothesis, Lemma \ref{lem:embedding-in-cosmo} implies that we embed within the cosmological horizon and Lemma \ref{lem: positivity} implies $E_\lambda \geq 0$. The rigidity claim is obtained as in Lemma \ref{lem: positivity} except $C$ ranges from $0 < C \leq \frac{2\sqrt{6}}{3\pi}$.  
\end{proof}

\begin{proof}[Proof of Corollary \ref{cor: riem}]
We let $\lambda\geq 0$ be any constant such that $H > 2\lambda$ and so
\begin{equation}
    E_{\LY} = \frac{1}{8\pi}\int_\Sigma (H_0 -\sqrt{H^2-4\lambda^2} )d\mu
\end{equation}
is defined. We aim to find a constant $\lambda_{\max}$ that does not depend on $\lambda$. For some constant $E\geq 0$ we define
 \begin{equation}\label{eq:def-lambdamax2}
    \a(E) = 
        \begin{cases}
            \frac{2\sqrt6}{3\pi} \sqrt{K_{\min}} & \text{if }  \sqrt{K_{\min}} < \frac{E}{r_\S^2}\\
            \min \left\lbrace \sqrt{\frac{E}{r_\S^2} \left( 2\sqrt{K_{\min}}-\frac{E}{r_\S^2} \right)} , \frac{2\sqrt6}{3\pi} \sqrt{K_{\min}}  \right\rbrace & \text{if }  \sqrt{K_{\min}} \geq \frac{E}{r_\S^2}
        \end{cases}.
    \end{equation}
We observe that $\a(E_{\LY})$\footnote{Note that this is what we called $\a$ given in \eqref{eq:def-lambdamax}.} is $\lambda$-dependent, while $\lambda_{\max} = \a(E_{\BY})$ is not.
We claim that $\a(E_{\LY}) \geq \a(E_{\BY})$. For these values of $\lambda$, note that $E_{\LY} \geq E_{\BY}$.
If $\sqrt{K_{\min}} < \frac{E_{\LY}}{r_\Sigma^2}$, then the claim follows immediately from the definition \eqref{eq:def-lambdamax2}. 
Suppose instead that 
\begin{equation}\label{eq:condition cor}
    \sqrt{K_{\min}} \geq \frac{E_{\LY}}{r_\Sigma^2}.
\end{equation}
As in the proof of Theorem \ref{thm: main}, we define
\begin{equation}
    \a_0(E) = \sqrt{\frac{E}{r_\Sigma^2}\left(2\sqrt{K_{\min}} - \frac{E}{r_\Sigma^2}\right)}.
\end{equation}
Then a simple computation shows
\begin{equation}
\begin{split}
    \a_0^2(E_{\LY}) - \a_0^2(E_{\BY}) 
    & = \frac{2 \sqrt{K_{\min}}(E_{\LY} - E_{\BY})}{r^2_\S} -  \frac{E_{\LY}^2 - E_{\BY}^2}{r^4_\S}  \geq \frac{(E_{\LY} - E_{\BY})^2}{r_\Sigma^4} 
    \geq 0,
\end{split}
\end{equation}
where we used \eqref{eq:condition cor}. This establishes the claim.
Therefore, if we choose any $\lambda \in [0, \l_{\max}]$ where $\l_{\max} = \a(E_{\BY})$, we have $\l \leq \a_\l$ and so by Theorem \ref{thm: main}, the energy $E_\lambda$ is nonnegative. Lastly, if $\l_{\max}>0$, then equation \eqref{eq:def-lambdamax2} implies $\alpha(E_{\BY})>0$ and so $E_{\BY}>0$. In this case, an argument as in \eqref{eq: E_lambda geq 0} shows that $E_\l>0$ for all $0 < \lambda \leq \lambda_{\rm max}$.
\end{proof}


\section{Examples}\label{sect: example}

In this section we present some examples of solutions to the Einstein equations with positive cosmological constant that are foliated by umbilic slices.

\subsection{Schwarzschild--de Sitter space}\label{subsect: SchwDeS}
\:
\medskip

In this section, we review the Schwarzschild--de Sitter space, a $\Lambda$-vacuum solution to the Einstein equations. First, we show how to express the metric in expanding coordinates. Then we study the positivity of energy $E_\lambda$ along coordinate spheres. Lastly we determine $\alpha_\lambda$.

Fix a mass parameter $ m > 0 $. 
Let $ \Lambda > 0$ denote the cosmological constant, and set as usual $\lambda = \sqrt{\tfrac{\Lambda}{3}}.$

Recall that the Schwarzschild--de Sitter metric in static coordinates is given by
\begin{equation}
g = -\left(1 - \frac{2m}{r} - \lambda^2 r^2\right)\,d T^{2}
  + \frac{d r^{2}}{1 - \frac{2m}{r} - \lambda^2 r^2}
  + r^{2}\,d \Omega^2,
\end{equation}
where $(S^2, d \Omega^2)$ denotes the unit round two-sphere.

If we perform the coordinate change
\begin{equation}
\label{eq:FLRW-change}
        T = t + \Phi(r), \qquad \Phi'(r) = \frac{\l r}{\left(1 - \frac{2m}{r} - \lambda^2 r^2\right) \sqrt{1 - \frac{2m}{r}}} ,
\end{equation}
then we obtain the metric in \emph{expanding coordinates}:
\begin{equation}\label{eq: SdS metric}
g \,=\, -\left(1 - \frac{2m}{r} - \lambda^2 r^2\right)\, d t^2 
      - \frac{2\lambda r}{\sqrt{1 - \tfrac{2m}{r}}}\, d r\,d t 
      + \frac{d r^2}{1 - \tfrac{2m}{r}} 
      + r^2\, d \Omega^2.
\end{equation}
Recognize that when $m = 0$ \eqref{eq: SdS metric} is just the de Sitter metric \eqref{eq: flat expanding de sitter} in flat-expanding coordinates.

Let $M_t$ denote a hypersurface of constant $t$ and $g$ its induced metric. If we denote the second fundamental form of $M_t$ by $k$, following our convention \eqref{eq: defi k}, 
a straightforward computation shows umbilicity:
\begin{equation}
k=\l g.
\end{equation}

Consider a large coordinate sphere $\S_R = \lbrace r=R \rbrace$. Let $H$ denote its mean curvature within $M_t$ and $H_0$ its mean curvature within Euclidean space via the Weyl embedding theorem. Assume $R$ is chosen such that  
\begin{equation}\label{eq: ass SchwDeS}
    H = \frac{2}{R}\sqrt{1-\frac{2m}{R}} > 2\l , \qquad H_0 =\frac{2}{R} > 2\l.
\end{equation}
We have

\begin{equation}
\begin{split}
    E_\lambda (\S_R) &= \frac{R^2}{2} \left( \sqrt{\frac{4}{R^2} - 4\lambda^2} \,-\, \sqrt{\frac{4}{R^2}\left(  1 - \frac{2m}{R} \right) - 4\lambda^2} \right) \\
    &> E_{\BY}(\S_R) \\
    &> 0.
\end{split}
\end{equation}
Therefore, the energy is positive for all the values of $\l$ for which it is well defined.

Note now that since $K = \frac{1}{R^2}$,
\begin{equation}
\begin{split}
\frac{E_{\LY}(\S_R)}{R^2} &= \frac{1}{R} \left(1 \,-\, \sqrt{  1 - \frac{2m}{R} - \lambda^2 R^2 } \right) <   \sqrt{K_{\min}}.
\end{split}
\end{equation}
Hence,

\begin{equation}\label{eq: alpha in SDS}
    \alpha_\lambda = 
        \min \left\lbrace \sqrt{\l^2 + \frac{2m}{R^3} } , \frac{2\sqrt6}{3\pi R}  \right\rbrace .
    \end{equation}

Note that the second expression in the definition of $\alpha_\lambda$ was chosen to guarantee, via Lemma \ref{lem:embedding-in-cosmo}, that the embedding of $\Sigma_R$ lies inside the cosmological horizon $r = \tfrac{1}{\l}$. This requirement, however, is not sharp: in the present setting it suffices to have $R \leq \frac{1}{\lambda}$, which is ensured by \eqref{eq: ass SchwDeS}. Thus only the first expression in \ref{eq: alpha in SDS} needs to be considered. Clearly,
\begin{equation}
    \l \leq \sqrt{\l^2 + \frac{2m}{R^3}. } 
\end{equation}
Therefore $E_\lambda > 0$ for all $\lambda$ whenever \eqref{eq: ass SchwDeS} holds.

The results obtained in this section for the exact Schwarzschild de Sitter case, pose a physically meaningful stability question. That is, whether a similar statement would hold for perturbations of this example. By the above computations, this would appear to be related to an associated stability analysis for the Weyl embedding theorem, a problem which could be of interest on its own right.

\subsection{Other $\Lambda$-vacuum examples}\label{subsect: SchoenEx}
\:
\medskip

In the previous section, we applied Corollary \ref{cor: riem} to the  Schwarzschild--de Sitter space, which is one example of a $\Lambda$-vacuum spacetime. In this section we show that there are many $\Lambda$-vacuum spacetimes where Corollary \ref{cor: riem} applies.

Let $V$ be a closed Riemannian manifold with positive scalar curvature. Fix $p \in V$ and set $M := V \setminus \{p\}$. The Green's function at $p$ for the conformal Laplacian on $V$ exists and is strictly positive. Consequently there is a scalar flat metric $g$ on $M$ which is asymptotically flat with $p$ representing the point at infinity, see e.g., \cite{Schoen1989,LeeParker1987}. Since $R_g = 0$, the initial data set $(M,g, k = \lambda g)$ satisfies the $\Lambda$-vacuum Einstein constraint equations with $\Lambda = 3\lambda^2$. In the asymptotically flat coordinates $x, y,z$, let $\Sigma_r$ be the sphere given by $r^2 = x^2 + y^2 + z^2$. For $r$ large enough, $\Sigma_r$ has positive Gauss and mean curvatures and bounds a compact region $\Omega_r$. For this sufficiently large $r$, let $\lambda_{\rm max}$ be given by Corollary \ref{cor: riem}. Then for all $\lambda \in [0, \lambda_{\rm max}]$,  $E_{\lambda}$ of $(\Omega_r, g, k = \lambda g)$ is well defined and nonnegative.

\section{Discussion and conclusion}\label{sect: conclusions}

\subsection{Dependence on the embedding}\label{sect: dep on embed}
\:
\medskip

We recall that the energy $E_\lambda$ depends crucially on the choice of isometric embedding used to define the reference term $\vec{H}_0$ (cf. \eqref{eq: defi mcv}), and hence on the choice of foliation of the reference de Sitter space. In our construction, we use the Weyl embedding theorem to embed $\Sigma$ isometrically into $\mathbb{R}^3$, which in turn sits as a umbilical hypersurface in de Sitter space. Had we chosen a different foliation, the associated reference embedding would change, and consequently the value of the energy $E_\lambda$ would, in general, be different.
 A similar phenomenon already appears in the Liu--Yau quasi-local energy, as clearly illustrated by the examples of O'Murchadha, Szabados, and Tod~\cite{OMurchadhaSzabadosTod2003}. They considered spacelike $2$-surfaces lying on null cones in flat Minkowski space and showed that, despite the fact that the ambient spacetime is globally flat, the Liu--Yau energy of such surfaces can be strictly positive. These examples demonstrate that the value of the quasi-local energy is sensitive to the chosen isometric embedding, rather than being determined solely by the intrinsic geometry of the surface. The same mechanism applies in our setting, as we now describe.

We follow the construction in~\cite{CanovasDeLaFuentePalomo2021}. Let $\Psi \colon \Sigma \to \mathcal{L}_{\mathscr{M}_0}$ be an immersion, where $\Sigma$ is a surface as in Theorem~\ref{thm: main} and $\mathcal{L}_{\mathscr{M}_0}$ denotes the light cone of a point in de~Sitter space with cosmological constant $\Lambda = 3\lambda^2$. By~\cite[Corollary~4.2]{CanovasDeLaFuentePalomo2021}, the Gauss curvature $K$ of $\Sigma$ is given by
\begin{equation}\label{eq:CFP}
    K = \frac14 \lvert \vec{H}_\Psi \rvert^2 + \lambda^2,
\end{equation}
where $\vec{H}_\Psi$ denotes the mean curvature vector of the image $\Psi(\Sigma)$.

On the other hand, $\Sigma$ can also be isometrically embedded into $\mathbb{R}^3$ via the Weyl embedding. Combining \eqref{eq:CFP} with the Gauss equation (see \eqref{eq: H_0 bound}), we obtain 
\begin{equation}\label{eq:different emb}
    H_0^2 - 4\lambda^2 - \lvert \vec{H}_\Psi \rvert^2
    = 2 \lvert \tf{A_0} \rvert^2 \ge 0,
\end{equation}
where $H_0$ and {$\tf{A}_0$} denote the mean curvature and trace-free part of the second fundamental form of the Weyl embedding, respectively. Note that equality holds if and only if $\Sigma$ is round \cite[Proposition 4.36]{ONeill1983}. This example makes explicit that the quasi-local energy $E_\lambda$ depends on the choice of isometric embedding used to define the reference term: replacing $\lvert H_0 \rvert$ in \eqref{eq: defi mcv} with $\lvert H_\Psi \rvert$ leads, in general, to a different value of $E_\lambda$. Choosing a spacelike graph of a sphere in the light cone, as in \cite{OMurchadhaSzabadosTod2003}, shows that,  one may obtain examples in which the value of the energy is strictly positive, specifically whenever the graph is not round.

These considerations suggest that a further promising direction is the development of an optimal embedding formulation. In the case of Minkowski space, Wang and Yau~\cite{WangYau1,WangYau2} introduced a variational quasi-local energy defined by minimizing over admissible isometric embeddings into Minkowski space. Constructing an analogous variational framework in the de~Sitter setting would provide a natural and geometrically consistent extension of our approach.

\subsection{Conclusion}\:

\medskip

In this work, we addressed the problem of defining energy within the framework of an expanding initial data set with a positive cosmological constant. In particular,  the presence of a cosmological horizon in the background de~Sitter space constitutes an obstruction to a global notion of energy, thereby motivating us to use quasi-local definitions instead. 
Inspired from the Liu--Yau construction in the asymptotically flat setting, we adapted their definition to the flat-expanding patch of de Sitter space:
\[
E_\lambda(\Sigma)
    = \frac{1}{8\pi} \int_{\Sigma}
      \left(
        \sqrt{H_0^2 - 4\lambda^2}
        - \sqrt{H^2 - (\operatorname{tr}_{\Sigma} K)^2}
      \right)
      d\mu,
\]
where $\lambda = \sqrt{\Lambda/3}$ denotes the expansion parameter associated with the cosmological constant $\Lambda > 0$. We established sufficient conditions guaranteeing that the isometric embedding of $\Sigma$ lies entirely within the cosmological horizon $r < 1/\lambda$, and proved positivity of $E_\lambda$ for a broad class of initial data sets, including the perfectly umbilic case, $k = \lambda g$, for all sufficiently small $\lambda$.

Several open questions naturally arise from this framework. A first one concerns rigidity: by analogy with Theorem \ref{thm: LY}, one expects that the vanishing of the quasi-local energy should should imply that the data is isometric to a domain in the flat-expanding patch of de~Sitter space. Formally, one may conjecture that
\[
E_\lambda(\Sigma) = 0 
\quad \Longrightarrow \quad 
(\Omega,g,k) \text{ is a domain in the flat expanding de~Sitter spacetime.}
\]
So far, we have obtained only a partial rigidity result: in the rigid case, the boundary $\Sigma = \partial \Omega$ must be round, and this occurs precisely at the critical value $\lambda = \alpha_\lambda$. It is an open question whether this rigidity can be upgraded to the full rigidity conjecture stated above.

Another open problem concerns the dependence of the quasi-local energy on the parameter $\lambda$. As discussed in Theorem~\ref{thm: main}, our result can be interpreted as identifying a range of admissible values of $\lambda$ for a fixed initial data set $(\Omega, g, k)$. It would be interesting to further investigate how the parameter $\lambda$ interacts with the intrinsic and extrinsic geometry of $\Sigma$.

\section*{Declarations}
\subsection*{Conflict of interest} The authors state that there is no Conflict of interest. No data was collected or analysed as part of this project.
\subsection*{Open Access} This article is licensed under a Creative Commons Attribution 4.0 International License, which
permits use, sharing, adaptation, distribution and reproduction in any medium or format, as long as you give appropriate credit to the original author(s) and the source, provide a link to the Creative Commons licence, and indicate if changes were made. The images or other third party material in this article are included in the article’s Creative Commons licence, unless indicated otherwise in a credit line to the material. If material is not included in the article’s Creative Commons licence and your intended use is not permitted by statutory regulation or exceeds the permitted use, you will need to obtain permission directly from the copyright holder.
To view a copy of this licence, visit http://creativecommons.org/licenses/by/4.0/.

\section*{Acknowledgments}
The research of EL was supported by Carlsberg Foundation CF21-0680 and, together with AP, by Danmarks Grundforskningsfond CPH-GEOTOP-DNRF151. RA would like to thank the Alexander von Humboldt foundation for partial financial support. 
This project took shape while EL and AP were participants in the program 
\emph{``New Frontiers in Curvature: Flows, General Relativity, Minimal Submanifolds, and Symmetry''} 
at the Simons Laufer Mathematical Sciences Institute (SLMath), supported by Grant No.~DMS-1928930 during the Fall~2024 semester. 
We are grateful for the many enlightening discussions with experts during the program, and in particular to Dan~Lee for his insightful conversations. 
EL and AP also acknowledge the support of the Simons Center for Geometry and Physics, where parts of this work were carried out during the program 
\emph{``Geometry and Convergence in Mathematical General Relativity''}. 
Lastly, we are deeply grateful to Pengzi Miao for valuable discussions and suggestions, especially regarding Lemma~\ref{lem:embedding-in-cosmo}.

\bibliographystyle{alpha}
\bibliography{bibliography}

\end{document}